\documentclass[a4paper, 10pt]{scrartcl}

\setcapindent{1em}
\usepackage[table,dvipsnames]{xcolor}
\definecolor{UOSGray}{RGB}{201,208,209}

\usepackage{comment}
\usepackage[USenglish]{babel}
\usepackage[utf8]{inputenc}

\usepackage[left=2.0cm, right=2.0cm, top=2.5cm, bottom=2.5cm]{geometry}

\usepackage{csquotes}
\usepackage{mathtools}
\usepackage{amsmath}
\usepackage{amsfonts}
\usepackage{amssymb}
\usepackage{amsthm}
\usepackage{mathtools}
\usepackage{bbm}
\usepackage{booktabs} %
\usepackage{nicematrix}
\usepackage{xspace}
\usepackage{multicol}
\usepackage{relsize}
\usepackage{subcaption}
\usepackage{algorithm}
\usepackage[noend]{algpseudocode}
\usepackage{arydshln}
\usepackage{tabularx}
\usepackage{cite}
\usepackage{hyperref}
\usepackage{multirow}
\usepackage{accents}

\usepackage{cite}
\usepackage{cleveref}

\newtheorem{theorem}{Theorem}%
\newtheorem{lemma}[theorem]{Lemma}%
\newtheorem{corollary}[theorem]{Corollary}%
\newtheorem{observation}[theorem]{Observation}%
\newtheorem{definition}[theorem]{Definition}%
\newtheorem{example}[theorem]{Example}%

\newcommand{\opt}{\ensuremath{\text{OPT}}}

\newcommand{\poly}{\ensuremath{\text{poly}}}
\newcommand{\demand}{\ensuremath{\delta}}
\newcommand{\maxDemand}{\ensuremath{\bar{\delta}}}
\newcommand{\maxDeg}{\ensuremath{\Delta}}
\newcommand{\maxLength}{\ensuremath{\bar{\ell}}}
\newcommand{\maxWeight}{\ensuremath{\bar{w}}}
\newcommand{\myLabeling}{\ensuremath{\pi}}
\newcommand{\lpTime}{\ensuremath{\text{LP}}}
\newcommand{\maxLB}{\ensuremath{\hat{w}}} %

\newcommand{\numCuts}{\ensuremath{\mu}}
\newcommand{\maxNumCuts}{\ensuremath{\bar{\numCuts}}}
\newcommand{\extension}[1]{\ensuremath{#1^{\maxDemand}}}
\newcommand{\myRestriction}[1]{\ensuremath{#1[uv]}}
\newcommand{\weightRestriction}[2]{\ensuremath{#1[#2]}}
\newcommand{\myDeg}{\ensuremath{\text{deg}}}
\newcommand{\dist}[2]{\ensuremath{d_{#1}^{#2}}}
\newcommand{\distL}[1]{\ensuremath{d_{#1}^{\ell}}}

\newcommand{\highE}{\ensuremath{E^{>}}}
\newcommand{\lowE}{\ensuremath{E[\maxLB]}}
\newcommand{\lowG}{\ensuremath{\weightRestriction{G}{\maxLB}}}

\newcommand{\metK}{\ensuremath{K'}}
\newcommand{\wMST}{\ensuremath{w(\mathit{MST}(G,w))}}

\newcommand{\GrigNumVars}{\ensuremath{|K|^5m^8\maxLength^8}}
\newcommand{\GrigNumCons}{\ensuremath{|K|^3m^4\maxLength^4}}
\newcommand{\GrigEnc}{\ensuremath{\max\{\ell, w\}}}
\newcommand{\GrigNumTrees}{\ensuremath{n\log(m \maxWeight)}}
\newcommand{\GrigLP}{\ensuremath{\lpTime(\GrigNumVars, \GrigNumCons, \GrigEnc)}}

\newcommand{\myNumVars}{\ensuremath{|K| m \maxDemand}}
\newcommand{\myNumCons}{\ensuremath{|K| (m+n\maxDemand)}}
\newcommand{\myEnc}{\ensuremath{w}}
\newcommand{\myLP}{\ensuremath{\lpTime(\myNumVars, \myNumCons, \myEnc)}}
\newcommand{\myLPUB}{\ensuremath{\lpTime(|K| m^2 \maxLength, |K| m^2 \maxLength, \myEnc)}}

\newcommand{\myAGTime}{\ensuremath{|K| (m + n\log n)\log m}}
\newcommand{\mySGTime}{\ensuremath{|K| (m + n\log n)}}
\newcommand{\myRRratio}{\ensuremath{n \maxNumCuts |K|}}

\newcommand{\MCF}{\text{\normalfont(MCF)}\xspace}

\newcommand{\generic}{decoupled\xspace}
\newcommand{\flexible}{freeform\xspace}
\newcommand{\freeform}{\flexible} %
\newcommand{\uniform}{coupled\xspace}

\newcommand{\GR}{\textsc{Greedy}\xspace}
\newcommand{\AG}{\textsc{AugmentedGreedy}\xspace}

\newcommand{\RR}{\textsc{RandomizedRounding}\xspace}

\newcommand{\Prob}{\ensuremath{\mathrm{Pr}}}

\newcommand\mc[1]{\multicolumn{2}{c|}{#1}} %

\DeclareMathSymbol{\widetildesym}{\mathord}{largesymbols}{"65}
    \newcommand\lowerwidetildesym{%
      \text{\smash{\raisebox{-1.3ex}{%
        $\widetildesym$}}}}

\newcommand\parwidetilde[1]{%
   \mathchoice
      {\accentset{\displaystyle\scalebox{.3}{(}\lowerwidetildesym\scalebox{.3}{)}}{#1}}
         {\accentset{\textstyle\scalebox{.3}{(}\lowerwidetildesym\scalebox{.3}{)}}{#1}}
         {\accentset{\scriptstyle\scalebox{.3}{(}\lowerwidetildesym\scalebox{.3}{)}}{#1}}
         {\accentset{\scriptscriptstyle\scalebox{.3}{(}\lowerwidetildesym\scalebox{.3}{)}}{#1}}
   }

\newcommand{\paraBigOTilde}{\ensuremath{\parwidetilde{\mathcal{O}}}}

\newcommand{\newStuff}[1]{{#1}}

\usepackage{tikz}
\usetikzlibrary{arrows,shapes,backgrounds,calc,positioning,arrows.meta,patterns}
\definecolor{myLightGrey}{RGB}{151,158,159}
\definecolor{myVeryLightGrey}{RGB}{191,198,199}

\tikzstyle{vertex}=[
    circle,
    fill=white, %
    draw=black,    %
    line width=1pt, %
    minimum size=20pt,
    inner sep=0pt
]

\tikzstyle{greyVertex}=[
    circle,
    fill=myVeryLightGrey, %
    draw=black,    %
    line width=1pt, %
    minimum size=20pt,
    inner sep=0pt
]

\tikzstyle{edge} = [draw,thick,-]
\tikzstyle{arc} = [draw,thick,->]

\newcommand{\includeTikzGraph}[1]{%
    \csname #1\endcsname
}



\title{Simple Approximations for\\ General Spanner Problems} %

\usepackage{orcidlink}

\author{Fritz Bökler \orcidlink{0000-0002-7950-6965} \\ \small fritz.boekler@uos.de \and Markus Chimani \orcidlink{0000-0002-4681-5550} \\\small markus.chimani@uos.de \and Henning Jasper\footnote{Corresponding author}~ \orcidlink{0000-0002-9821-8600}\\\small henning.jasper@uos.de}

\date{\small %
Institute of Computer Science, Osnabrück University, Germany}

\begin{document}

\maketitle

\begin{abstract}
Consider a (possibly directed) graph with $n$ nodes and $m$ edges,
independent edge weights and lengths, and arbitrary distance demands for node pairs.
The \emph{spanner problem} asks for a minimum-weight subgraph that satisfies these %
demands via sufficiently short paths w.r.t.\ the edge lengths.
For \emph{multiplicative $\alpha$-spanners} (where %
demands equal $\alpha$ times the original distances) and assuming that 
each edge's weight equals its length, 
the simple \GR heuristic by Althöfer et al.\ (1993) is known to yield strong solutions, both in theory and practice. %
To obtain guarantees in more general settings, recent approximations typically abandon this simplicity and practicality.
Still, so far, there is no known non-trivial approximation algorithm for the spanner problem in its most general form.

We provide two surprisingly simple approximations algorithms.
In general, our \AG achieves the first unconditional approximation ratio of $m$, which is non-trivial due to the independence of weights and lengths. %
Crucially, it maintains all size and weight guarantees \GR is known for, i.e., in the aforementioned multiplicative $\alpha$-spanner scenario and even for \emph{additive $+\beta$-spanners}.
Further, it 
generalizes some of these size guarantees to derive new weight guarantees. %

Our second approach, \RR, establishes a graph transformation that allows a simple rounding scheme over a standard multicommodity flow LP.
It yields an $\mathcal{O}(n \log n)$-approx\-imation, assuming integer lengths and polynomially bounded distance demands.
The only other known approximation guarantee in this general setting requires several complex subalgorithms and analyses, yet we match it up to a factor of $\mathcal{O}(n^{1/5-\varepsilon})$ using standard tools.
Further, on bounded-degree graphs, we yield the first $\mathcal{O}(\log n)$ approximation ratio for constant-bounded distance demands (beyond undirected multiplicative $2$-spanners in unit-length graphs).
\end{abstract}

\newpage

\section{Introduction}
Let $G=(V,E)$ be a graph with $n$ nodes and $m$ edges.
For $u,v \in V$, the distance $d_G^\ell(u, v)$ is the length of a shortest path from $u$ to $v$ in $G$, subject to %
lengths $\ell \colon E \rightarrow \mathbb{Q}_{>0}$.
A \emph{spanner} is a subgraph $H=(V,E')$, with $E'\subseteq E$, that preserves these distances up to some error; see Ahmed et al.~\cite{ahmed2020} for an overview. %
Spanner variations differ in their \emph{distance demand} function $\demand \colon K \rightarrow \mathbb{Q}_{>0}$ that maps \emph{terminal pairs} $K \subseteq V \times V$ to maximum allowed distances in $H$; $H$ is feasible if and only if $d_H^\ell(u,v) \leq  \demand(u,v)$ for all $(u,v) \in K$. Most generally, we have:
\begin{definition}[(Decoupled Freeform) Spanner Problem]\label{def:GSP}
    Given a connected, simple directed or undirected graph $G=(V,E)$, terminal pairs $K \subseteq V \times V$, non-negative weights $w \colon E \rightarrow \mathbb{Q}_{\geq 0}$, positive lengths $\ell \colon E \rightarrow \mathbb{Q}_{>0}$, and positive distance demands $\demand \colon K \rightarrow \mathbb{Q}_{>0}$.
    Find a feasible spanner $H=(V,E')$ of $G$ with minimum total weight $w(H)\coloneqq w(E')\coloneqq \sum_{e\in E'} w(e)$.
\end{definition}

Spanners in general graphs were first introduced %
to study communication networks~\cite{peleg1989, Ullman1989}.
They since found application in, e.g., 
network design with bounded distances~\cite{li1992,dodis1999,chimani2014network}, approximate distance oracles~\cite{peleg1999proximity, thorup2005}, graph drawing~\cite{wallinger2023, archambault2024}, or passenger assignment~\cite{heinrich2023}. 

Arguably, \Cref{def:GSP} is the broadest interpretation of the classic spanner problem.
In many applications, the weight and length of an edge are completely independent, we say \emph{\generic}.
However, the literature often considers simplifying assumptions, cf.\ \Cref{tab:ProblemNames}(left): 
In \emph{basic} (a.k.a.\ minimum~\cite{gomez2023} or unweighted~\cite{elkin2007}) spanner problems, all edges have unit weight and length, i.e., $w(e)=\ell(e) = 1$, for all $e \in E$.
In \emph{\uniform} (a.k.a.\ light-~\cite{chimani2014network} or minimum-weight~\cite{Sigurd2004}) spanner problems, each edge's weight equals its length, i.e., $ w(e) = \ell(e) > 0$, for all $e \in E$.
\emph{Unit-length} or \emph{unit-weight} spanner problems assume $\ell(e) = 1$, or $w(e) = 1$, for all $e\in E$, respectively.
Unit-weight variations thus minimize the \emph{size} $|E'|$ of spanners.

In general, every terminal pair $(u,v) \in K$ has an arbitrary distance demand $\demand(u,v)$. %
We call this \emph{\flexible} spanner problems.
However, many spanner variations assume a consistent global form, cf. \Cref{tab:ProblemNames}(right):
In \emph{linear $(\alpha, \beta)$-spanners}, for $\alpha \geq 1$ and $\beta \geq 0$  (typically small constants), the distance demands are $\demand(u,v)=\alpha \cdot d_G^\ell(u, v) + \beta$, for all $(u,v) \in K = V \times V$. 
Notable special cases include \emph{multiplicative $\alpha$-spanners} (or $(\alpha,0)$-spanners), \emph{additive $+\beta$-spanners} (or $(1, \beta)$-spanners) and \emph{distance preservers} (or $(1, 0)$-spanners).

\begin{table}[b]
    \centering
    \caption{Naming-conventions. (left) Restrictions on weights $w(e)$ and lengths $\ell(e)$, for $e\in E$. (right) Restrictions on distance demands $\demand(u,v)$, for $(u,v) \in K$.}
    \label{tab:ProblemNames}
        \phantom{.}\hfill
        \begin{tabular}{c|c|c}
         name & $w(e)$ & $\ell(e)$ \\
         \hline
         \generic & $\mathbb{Q}_{\geq0}$ &  $\mathbb{Q}_{>0}$  \\
         unit-length & $\mathbb{Q}_{\geq0}$ & $1$  \\
         unit-weight & $1$ & $\mathbb{Q}_{>0}$  \\
        \hline
         \uniform  & \multicolumn{2}{c}{$\mathbb{Q}_{>0} = \mathbb{Q}_{>0}$}\\
         basic & \multicolumn{2}{c}{$1 = 1$}  \\
        \end{tabular}
        \hfill\hfill
        \begin{tabular}{c|l}
            name & \multicolumn{1}{c}{$\demand(u,v)$}\\
            \hline
            \flexible spanner & \multicolumn{1}{c}{$\mathbb{Q}_{>0}$}\\
            linear $(\alpha, \beta)$-spanner & $\alpha \cdot d_G^\ell(u, v) + \beta$ \\
            multiplicative $\alpha$-spanner & $\alpha \cdot d_G^\ell(u, v)$\\
            additive $+\beta$-spanner & \phantom{$\alpha \cdot$} $d_G^\ell(u, v) + \beta$\\
            distance preserver & \phantom{$\alpha \cdot$} $d_G^\ell(u, v)$\\
        \end{tabular}
        \hfill\phantom{.} 
\end{table}

Unfortunately, already severely restricted variations are hard to solve, or at least approximate, in polynomial time (under mild complexity assumptions). %
The \emph{undirected basic multiplicative $\alpha$-spanner} problem is \textbf{NP}-hard~\cite{peleg1989}, even for degree bounded graphs~\cite{gomez2023}.
For $\alpha=2$, the best possible approximation ratio is $\Theta(\log n)$~\cite{kortsarz2001}; for $3 \leq \alpha \leq \log^{1-2\varepsilon} n$, it allows no ratio better than $2^{(\log^{1-\varepsilon} n)/\alpha}$, for any constant $\varepsilon>0$~\cite{dinitz2015}.
Similarly, for $\beta \geq 1$, the \emph{undirected basic additive $+\beta$-spanner} problem is \textbf{NP}-hard~\cite{liestman1993}, and allows no $2^{\log^{1-\varepsilon} n}/\beta^3$-approximation, for any constant $\varepsilon>0$~\cite{chlamtavc2020approximating}.
Further, the \emph{undirected \generic multiplicative $\alpha$-spanner} problem allows no  $2^{\log^{1-\varepsilon} n}$-approximation, for any $1 < \alpha \in \mathcal{O}(n^{1-\varepsilon'})$ and $\varepsilon, \varepsilon' \in (0,1)$%
, even for polynomially bounded weights and lengths~\cite{elkin2007}.
Clearly, the \generic \freeform\ spanner problem inherits all of these hardness results.

Still, many special cases have known approximation algorithms.
Most research focuses on multiplicative $\alpha$-spanners.
For basic $\alpha$-spanners with $\alpha = 2$~\cite{kortsarz1994} and $\alpha \in \{3,4\}$~\cite{berman2013,zhang2016}, there are $\mathcal{O}(\log(m/n))$- and $\widetilde{\mathcal{O}}(n^{1/3})$-approximations, respectively.
Thereby, $\widetilde{\mathcal{O}}$ suppresses polylogarithmic factors.
For unit-length $\alpha$-spanners, there are $\mathcal{O}(\log n)$-, $\widetilde{\mathcal{O}}(n^{1/2})$-, and $\mathcal{O}(n \log \alpha)$-approximations for $\alpha=2$~\cite{kortsarz2001}, $\alpha=3$~\cite{dinitz2011directed}, and arbitrary integer $\alpha$~\cite{dodis1999}, respectively.

Regarding non-multiplicative spanners, notable algorithms yield basic additive $+2$-, $+4$-, and $+6$-spanners with size in $\mathcal{O}(n^{3/2})$~\cite{knudsen2014, aingworth1999fast}, $\widetilde{\mathcal{O}}(n^{7/5})$~\cite{chechik2013new}, and $\mathcal{O}(n^{4/3})$~\cite{baswana2010additive, knudsen2014}, respectively.
Some of these guarantees hold in the unit-weight setting for relaxed additive errors
~\cite{elkin2019almost, ahmed2020weighted, ahmed2021additive, elkin2023improved}.
Other constructions focus on basic~\cite{elkin2004_1+epsilon, baswana2005new, elkin2005computing} or \uniform~\cite{ahmed2021additive, gitlitz2024lightweight} linear $(\alpha, \beta)$-spanners.

Many spanner algorithms form a union over several well-chosen paths between terminal pairs.
Arguably, the most straightforward approach is the \GR heuristic, 
which incrementally adds shortest paths to an initially empty spanner until it is feasible.
Famously, Althöfer et al.~\cite{althofer1993} apply this technique 
to the \emph{undirected \uniform multiplicative $\alpha$-spanner} problem, 
yielding spanners of lightness $\mathcal{O}(n/k)$ and size $\mathcal{O}(n^{1+1/k})$, for all $\alpha = 2k-1$ ($k \in \mathbb{N}_{\geq 1}$).
The \emph{lightness} $\frac{w(E')}{\wMST}$ %
upper bounds the approximation ratio, as it puts the spanner's weight in relation to the weight of a minimum-weight spanning tree $\mathit{MST}(G,w)$ of $G$ (a simple lower bound to the optimum).
By violating distance demands by a factor of $(1+\varepsilon)$, $\varepsilon>0$, \GR yields $(1+\varepsilon) (2k-1)$-spanners of lightness $\mathcal{O}(n^{1/k} (1+k/(\varepsilon^{1+1/k} \log k)))$~\cite{elkin2015light}.
In practice, \GR typically outperforms other non-exact approaches~\cite{chimani2022}, often achieving near-optimum weight~\cite{jasper2024}.
By Knudsen~\cite{knudsen2014}, \GR also yields basic additive $+2$-spanners of size $\mathcal{O}(n^{3/2})$. %

To achieve approximation guarantees for broader problem variations, \GR does not suffice.
Corresponding algorithms, e.g., ~\cite{dinitz2011directed,chlamtavc2020approximating,chimani2014network,grigorescu2023approximation}, select paths more intricately, sacrificing simplicity and speed.
Further, they are restricted to integer lengths%
\footnote{Newer publications say so explicitly; older ones may gloss over this restriction, arising from using an FPTAS to solve constrained shortest path (CSP) problems over integer lengths. Indeed, CSP with rational lengths is strongly \textbf{NP}-hard~\cite{wojtczak2018}.}.
The algorithms are randomized both w.r.t.\ approximation ratio and feasibility.
They tackle the LP relaxation of the \emph{Path-ILP}~\cite{Sigurd2004} for spanners as a subroutine.
Already solving the LP is notoriously difficult, as variables must be dynamically generated by repeatedly solving (\textbf{NP}-hard) \emph{constrained shortest path} problems~\cite{Garey1979}.
The latter is achieved by using an FPTAS~\cite{lorenz2001simple,horvath2018multi}, yielding a $(1+\varepsilon)$-approximation for the relaxation.
The spanner approximations then cover some of the terminal pairs by rounding the approximate fractional LP solution.
Covering the remaining terminal pairs using shortest path trees, \emph{junction trees}~\cite{chekuri2010approximation}, or \emph{shallow-light Steiner trees}~\cite{Garey1979} yields randomized approximations for directed unit-weight multiplicative $\alpha$-spanners with ratio $\widetilde{\mathcal{O}}(n^{2/3})$~\cite{dinitz2011directed}, basic \freeform\ spanners with ratio $\mathcal{O}(n^{3/5+\varepsilon})$~\cite{chlamtavc2020approximating}, or \generic multiplicative $(\alpha^2 + \alpha\varepsilon)$-spanners with at most $\mathcal{O}(n^{1/2+\varepsilon})$ times the weight of an optimum \generic $\alpha$-spanner~\cite{chimani2014network}, for $\varepsilon>0$, respectively.
Grigorescu, Kumar and Lin~\cite{grigorescu2023approximation} combine several of these ideas:
For directed \generic distance preservers with integer lengths and polynomially bounded distance demands, they present a randomized $\mathcal{O}(n^{1/2 + \varepsilon})$-approximation 
(albeit implicitly, the same guarantee was achieved before in~\cite{chimani2014network}).
Most importantly, they yield the first randomized $\widetilde{\mathcal{O}}(n^{4/5 + \varepsilon})$-approximation for the directed \generic \flexible spanner problem with integer lengths and polynomially bounded distance demands.
To attain a sublinear ratio for a problem of this generality, the algorithm employs several further complex and practically never-implemented subroutines.
It iterates over guesses for the optimum solution value; 
every iteration not only requires numerous approximations of Path-LP adaptations, but also several constrained shortest path trees~\cite{lorenz2001simple}, and approximations of \emph{minimum-weight distance-preserving junction trees}~\cite{chlamtavc2020approximating} via \emph{shallow trees}~\cite{garg2000polylogarithmic}.
The running time is not provided in $\mathcal{O}$-notation, but only argued to be generally polynomial; the expected weight guarantee is given in $\widetilde{\mathcal{O}}$-notation, neglecting polylogarithmic factors.

Currently, there are no practical approximations for \generic, let alone \generic \freeform, spanner problems.
New spanner approximations tend to focus on improving weight guarantees in $\paraBigOTilde$-notation rather than considering running times or practicality.
Yet, despite these efforts, an approximation for general \generic \freeform spanner problems that does not require polynomially bounded integer lengths has remained elusive so far.
Consequently, practitioners are prompted to force \generic instances into a \uniform format by discarding or combining weights and lengths (e.g.~\cite{archambault2024, heinrich2023}).
Instances are then often solved with \GR, which is easy to implement and understood to yield satisfactory results~\cite{chimani2022, jasper2024}.
Still, this practice poses significant issues, as it does not guarantee that solutions effectively address the original \generic problem.

\subparagraph{Contribution.} 
We present two simple approximation algorithms for both directed and undirected \generic \freeform\ spanner problems. 
Tables \ref{tab:myResults} and \ref{tab:AGResults} summarize our contributions.

In \Cref{sec:2PhaseGreedy}, we augment \GR to become the first unconditional (i.e., no upper bounds or domain restrictions for weights or lengths, etc.) $m$-approximation for all \generic \freeform\ spanner problems. The ratio $m$ is non-trivial due to the independence of weights and lengths.
Our \AG not only maintains \GR's running time up to a factor of $\mathcal{O}(\log m)$, but crucially retains all size and weight guarantees typically associated with standard \GR.
Further, it allows us to lift known size guarantees from several \uniform settings into \generic weight guarantees.

The best known approximation algorithm~\cite{grigorescu2023approximation} for (un)directed \generic \freeform\ spanners assumes integer lengths and polynomially bounded maximum distance demand $\maxDemand$.
It achieves ratio $\widetilde{\mathcal{O}}(n^{4/5 + \varepsilon})$, but relies on several complex subalgorithms with complicated analyses.
In \Cref{sec:RandomizedRounding}, for these instances, we provide a simple graph transformation enabling a \generic \freeform\ spanner approximation via randomized rounding of a standard multicommodity flow LP.
Let $\maxNumCuts=(\maxDemand + 2)^{n-2}$; our ratio of $\ln(\myRRratio)$ $(\in \mathcal{O}(n \log n))$ accounts for all constants and matches the best known ratio up to a factor of $\mathcal{O}(n^{1/5-\varepsilon})$.
Our approximation requires no specialized subroutines, and offers significantly stronger time guarantees as well as easy implementation and analysis.
Also, this simple analysis allows us to prove the first $\mathcal{O}(\log n)$-approximation ratio for (un)directed \generic \freeform\ spanner problems with constant-bounded distance demands and (out)degree (which are \textbf{NP}-hard~\cite{gomez2023}).
Previously, for constant-bounded demands 
only the special case of undirected unit-length multiplicative $2$-spanners 
(i.e., $\demand(u,v) = 2$ for all $\{u,v\} \in K=E$) 
was known to allow this strong ratio. %

\begin{table}
    \caption{Guarantees for the (un)directed decoupled freeform spanner problem. %
    $\lpTime(x,y,f)$ denotes the running time for solving LPs with $\Theta(x)$ variables and $\Theta(y)$ constraints, whose largest  numbers arise in the image of function $f$ (cf.\ \Cref{sec:preliminaries}). \enquote{---} marks inapplicability due to exponential running times. \enquote{$^{\$}$} denotes the solution is feasible w.h.p and requires integer lengths. Let $\maxWeight, \maxLength,  \maxDemand, \maxDeg$ denote the maximum weight, length, distance demand, and (out)degree of an instance, respectively. Let $\varepsilon >0$, and $\maxNumCuts=(\maxDemand + 2)^{n-2}$.
    }
    \label{tab:myResults}
    \resizebox{\textwidth}{!}{
    \begin{tabular}[c]{c|c|c| r @{${\ } \mathcal{O} {}$} l |c}
    \multicolumn{2}{c|}{}& best known  & \multicolumn{3}{c}{our contribution} \\
    \multicolumn{2}{c|}{} & (Grigorescu et al.~\cite{grigorescu2023approximation}) &  \mc{\RR} & \AG \\
     \hline
     \multicolumn{2}{c|}{} & $\poly(n, |K|, \log \maxWeight, \maxLength, \varepsilon^{-1})$, & \mc{} & \\
     \multicolumn{2}{c|}{running time (in $\mathcal{O}$)} & $\Omega \big( \GrigNumTrees \cdot$ & \mc{$\myLP$} & $\myAGTime$\\
     \multicolumn{2}{c|}{} & $\GrigLP \big)$ & \mc{}& \\

     \hline
      &  general & --- & \mc{---} & $m$  \\
     approx. & $\maxDemand \in \mathcal{O}(\text{poly(n)})$ & $\widetilde{\mathcal{O}}(n^{4/5 + \varepsilon})$ $^{\$}$  & $\ln(\myRRratio) \in$ & $(n \log n)$ $^{\$}$ &  $m$ \\
     ratio & $\maxDemand \in \mathcal{O}(1)$ &  $\widetilde{\mathcal{O}}(n^{4/5 + \varepsilon})$ $^{\$}$ &  & $(n)$ $^{\$}$ &  $m$ \\
     & $\maxDemand, \maxDeg \in \mathcal{O}(1)$ & $\widetilde{\mathcal{O}}(n^{4/5 + \varepsilon})$ $^{\$}$ &  &$(\log n)$ $^{\$}$ &  $m \in \mathcal{O}(n)$ \\
    \hline

    \end{tabular}
    }
\end{table}

\begin{table}
    \caption{Overview of \GR vs.\ \AG.
    For \uniform multiplicative and basic additive spanners, they obtain the same weight and size guarantees (\Cref{cor:2PhaseGreedyRetention,cor:AdaptedGreedy+beta}).
    Let $k \in \mathbb{N}_{\geq1}$.
    \texttt{GEO} denotes the geometric setting;
    $\highE \subseteq E$ denotes \enquote{high-weight} edges; see \Cref{sec:2PhaseGreedy} for details.
    }
    \label{tab:AGResults}
    \resizebox{\textwidth}{!}{
    \centering
    \begin{tabular}{c|c|c |c|c}
    spanner problem  & setting  &  
    \GR & \multicolumn{2}{c}{\AG} \\
    \cline{3-5}
    & & \multicolumn{2}{c|}{weight ratio guarantee} & size\\
    \hline
    \hline
    (un)directed \generic \freeform & --- & --- & $m$ & $m$ (trivial)\\
    \hline
     \multirow{2}{*}{undirected \generic multiplicative} & $\alpha=2k-1$, $|\highE| \in \mathcal{O}(\sqrt{n} \cdot n^{1/k})$& --- & $\mathcal{O}(n^{1+1/k})$ & \multirow{3}{*}{\shortstack[c]{\ \\same as\\ \GR}} \\
      &  \texttt{GEO}, $\alpha>1$, $|\highE| \in \mathcal{O}(\sqrt{n})$ & --- & $\mathcal{O}(n)$ & \\
     \cline{1-4}
     undirected unit-length additive & $2 \leq \beta \in \mathcal{O}(1)$ & --- & $\mathcal{O}(n^{3/2})$ & \\
    \end{tabular}
    }
\end{table}

\section{Preliminaries}\label{sec:preliminaries}
Generally, \emph{shortest} means minimum length. 
In undirected settings, unordered terminal pairs $K \subseteq \binom{V}{2}$ suffice.
With $\opt$ we denote the optimum solution value of the considered instance.
We define the maximum length $\maxLength \coloneqq \max_{e\in E} \ell(e)$, weight $\maxWeight \coloneqq \max_{e\in E} w(e)$, distance demand $\maxDemand \coloneqq \max_{(u,v) \in K} \demand(u,v)$ and %
(out)degree $\maxDeg\coloneqq \max_{q \in V} |\{(q,t) \mid (q,t) \in E\} |$.
We can assume w.l.o.g.\ $\maxLength \leq \maxDemand \leq n \maxLength$, as longer edges cannot be used to satisfy any distance demands and larger demands are trivially satisfied in connected graphs.
For a function~$f\colon E \rightarrow \mathbb{Q}$, we use the shorthands $f(u,v) \coloneqq f((u,v))$ for $(u,v) \in E$, $f(X) \coloneqq\sum_{e \in X} f(e)$ for $X \subseteq E$, and $f(G')\coloneqq f(E')$ for $G'=(V',E')\subseteq G$.
Let $\lpTime(x, y, f)$ denote the running time to solve an LP with $\Theta(x)$ variables and $\Theta(y)$ constraints;
thereby the maximum encoding-length of any coefficient does not exceed that of the largest rational in the image of function $f$.

Enforcing distance demands $\demand(u,v) = \alpha \cdot d_G^\ell(u,v)$, for all \emph{adjacent} node pairs $(u,v) \in E=K$ guarantees feasible multiplicative $\alpha$-spanners~\cite{peleg1989}.
Further, it is folklore that non-metric edges can be removed from \emph{\uniform} $\alpha$-spanner instances, as they cannot be part of an optimum solution.
However, in \generic settings, non-metric edges \emph{cannot} be removed:
Consider an undirected triangle, where two of the edges have unit weight and length, while the third has weight $0.5$ and (non-metric) length $3$.
The optimum $4$-spanner has weight $1.5$, 
while removing the third edge leaves us with a minimal solution value of $2$.
Yet, the following Observation %
allows us to disregard non-metric distance demands in \generic \freeform\ spanner problems.

In the \emph{demand graph} $D=(V,K)$, 
two nodes are adjacent if there is a distance demand between them. Then $d_D^\demand(u,v)$ is the distance between nodes $u$ and $v$ in $D$ w.r.t.\ the distance demands $\demand$.
We define the \emph{metric} terminal pairs as those whose distance demands (strictly) satisfy the triangle inequality, i.e., $\metK\coloneq \{ (u,v) \in K \mid d_{D-(u,v)}^\demand(u,v) > \demand(u,v)\}$. %
\begin{observation}\label{lem:metricTerminalPairs}
        A subgraph $H=(V,E')$ is a feasible (un)directed \generic \flexible spanner of $G=(V,E)$, if and only if $d_H^\ell(u,v) \leq \demand(u,v)$ %
        for all metric terminal pairs $(u,v)$. %
\end{observation}
\begin{proof}
        Necessity follows from $\metK \subseteq K$. 
        Conversely, assume $d_H^\ell(u,v) \leq \demand(u,v)$ for all $(u,v) \in \metK$.
        Let $(s,t) \in K \setminus \metK$. Let $P \subseteq K$ be a minimum $st$-path in $D$ w.r.t.\ $\demand$;
        w.l.o.g.\ $P \subseteq \metK$. Consequently, $d_H^\ell(s,t) \leq \sum_{(i,j) \in P} d_H^\ell(i,j) \leq \sum_{(i,j) \in P} \demand(i,j) \leq \demand(s,t)$.%
\end{proof}

\subparagraph{Dodis and Khanna's Algorithm.}
Going from unit-length to general \generic settings is non-trivial.
In~\cite{dodis1999}, Dodis and Khanna give a $\mathcal{O}(n \log \alpha)$-approximation for directed unit-length multiplicative $\alpha$-spanners.
They use a multicommodity flow LP, that, for every edge $e \in E$, sends flow through a layered graph.
Randomized rounding yields a subgraph $H$, such that 
(w.h.p.)\ $\distL{H}(u,v) \leq \alpha \cdot \distL{G}(u,v)$, for every $G$-adjacent node pair $(u,v) \in E$.

At the end of the paper, they briefly suggest how to generalize this algorithm to the \generic setting with polynomially bounded lengths.
However, their approach unfortunately does not work:
In short, they transform the given instance into a unit-length instance where demands are only given for some node pairs (in particular not all adjacent node pairs) and want to solve it using the previous algorithm.
However, that algorithm requires that distance demands are enforced for \emph{all adjacent} node pairs. Since this is not the case here, it cannot be expected to yield feasible solutions in general. See \Cref{appendix:Dodis} for an explicit example.

Nonetheless, our \RR algorithm in~\Cref{sec:RandomizedRounding} will leverage their idea of multicommodity flow in layered graphs to, in fact, yield approximations even for more general \generic \freeform (not only multiplicative) spanner problems.

\section{Augmented Greedy}\label{sec:2PhaseGreedy}

\begin{algorithm}[tb]
\caption{\GR.}\label{alg:StandardGreedy}
\begin{algorithmic}[1]
\Require graph $G=(V,E)$, terminal pairs $K$, lengths $\ell$, distance demands $\demand$
\State $H$ $\leftarrow$ $(V,\emptyset)$
\For{$(u,v) \in K$ in non-decreasing order of distance $d_G^\ell(u,v)$}
    \If{$d^\ell_H(u,v) > \demand(u,v)$}
        \State add shortest $uv$-path in $G$ (w.r.t.\ lengths $\ell$) to $H$ %
    \EndIf
\EndFor
\State \textbf{return} $H$
\end{algorithmic}
\end{algorithm}

Arguably, the easiest spanner algorithm is the \GR heuristic, shown in \Cref{alg:StandardGreedy}, which disregards weights: Start with an empty spanner $H=(V,\emptyset)$. %
Until $H$ is feasible, process terminal pairs $(u,v) \in K$ in non-decreasing order of distance in $G$.
If the current distance in $H$ is too long, add a shortest $uv$-path in $G$ to $H$.
Thereby, we break ties consistently to ensure unique shortest paths, e.g., via well-chosen perturbation or lexicographic ordering.
These shortest path computations for all $(u,v) \in K$ dominate the running time of $\mathcal{O}(\mySGTime)$,
where some cases allow refined analyses~\cite{althofer1993, chandra1992}.
Despite its simplicity, \GR offers strong guarantees for certain spanner problems.
For undirected \uniform multiplicative $\alpha$-spanners, it guarantees lightness $\mathcal{O}(n/k)$ and size $\mathcal{O}(n^{1+1/k})$~\cite{althofer1993}, for $\alpha = 2k-1$ ($k \in \mathbb{N}_{\geq 1}$).
The latter is tight assuming Erdős' girth conjecture~\cite{erdos1964}.
For $\alpha \geq \log n$ and $\alpha \geq (\log n)^2$ the lightness improves to $\mathcal{O}(\log n)$ and $\mathcal{O}(1)$, respectively~\cite{chandra1992}.
In the undirected \uniform \emph{geometric} setting, i.e., complete graphs with Euclidean  lengths, it guarantees lightness $\mathcal{O}(\log n)$ and size $\mathcal{O}(n)$, for all $\alpha>1$~\cite{chandra1992} (observe that this changes the co-domains of $w,\ell,\demand$ to $\mathbb{R}_{>0}$).
For spanner problems, where feasible solutions are undirected connected subgraphs that span all nodes, 
the weight of a minimum-weight spanning tree $\wMST$
lower bounds the weight of \emph{all} feasible spanners.
Hence, as guarantees~\cite{althofer1993, chandra1992} are absolute or w.r.t.\ $\wMST$, they analogously hold for undirected \uniform linear $(\alpha, \beta)$-spanners, for any $\beta \geq 0$.
Further, for undirected basic additive $+2$-spanners, \GR has (best possible~\cite{woodruff2006}) size guarantee $\mathcal{O}(n^{3/2})$~\cite{knudsen2014}.
However, in general, \GR has no known approximation guarantees for spanner problems with \generic weights and lengths or if distance demands do not take linear form.

Our \AG overcomes these limitations, see \Cref{alg:AdaptedGreedy}.
It runs in two consecutive phases:
Phase 1 achieves the first unconditional approximation ratio for all \freeform\ spanner problems, including \generic weights and lengths.
For this, it establishes a feasible spanner $G'$ that already attains an initial weight guarantee. %
Phase 2 uses \GR to further sparsify $G'$ to obtain the final spanner $H$.
While without phase 1 \GR does not achieve any weight guarantee in general, executing \GR after phase 1 ensures it still finds a feasible solution, maintains the weight ratio of the first phase, and achieves \GR's known guarantees at the same time. Crucially, \AG's total running time closely matches that of \GR.

\begin{algorithm}[tb]
\caption{\AG. Line 2 is only executed if the problem definition provides $\wMST \leq \opt$ (see text).}\label{alg:AdaptedGreedy}
\begin{algorithmic}[1]
\Require graph $G=(V,E)$, terminal pairs $K$, weights $w$, lengths $\ell$, distance demands $\demand$
\State binary search for $\maxLB = \min \{ w'\in\mathcal{W}: \weightRestriction{G}{w'} \text{ is a feasible spanner} \}$ \hfill $//$ phase 1
\State \textsf{\textbf{[}} $\maxLB$ $\leftarrow$ $\max \{\maxLB, \wMST\}$  \textsf{\textbf{]}}
\State construct $\lowG$ 
\State \textbf{return} {\GR}$(\lowG, K, \ell, \demand)$ \hfill $//$ phase 2
\end{algorithmic}
\end{algorithm}

The key idea in phase 1 is to determine an adequate lower bound $\maxLB$ to the optimum solution value $\opt$.
Let $\mathcal{W}$ be the set of distinct edge weights in the instance. For any $w'\in \mathcal{W}$, we define weight-restricted subgraphs $\weightRestriction{G}{w'}=(V,\weightRestriction{E}{w'})$ of $G$ with $\weightRestriction{E}{w'} \coloneqq \{e \in E\colon w(e) \leq w'\}$.
Checking whether some $\weightRestriction{G}{w'}$ constitutes a feasible spanner, i.e.,  $d^\ell_{\weightRestriction{G}{w'}}(u,v) \leq \demand(u,v)$ for all $(u,v) \in K$, is trivial.
Thus, using binary search over sorted $\mathcal{W}$, we can easily find the smallest weight $\maxLB$ such that $\weightRestriction{G}{\maxLB}$ is a feasible spanner.
Clearly, every feasible spanner contains an edge of weight at least~$\maxLB$ and $\maxLB \leq \opt$.
Further, if the problem definition ensures that a spanner is undirected and must be spanning and connected
(e.g., for undirected multiplicative $\alpha$-spanner problems), the weight of a minimum-weight spanning tree is a natural lower bound for $\opt$.
In such a case, we may further increase $\maxLB$ to that weight, overall knowing that any feasible solution will require a total weight of at least $\maxLB$. %
Phase 2 further sparsifies $\lowG$.
For all $(u,v) \in K$, to ensure $d^\ell_{H}(u,v) \leq \demand(u,v)$ in the final spanner $H$, distances in $\weightRestriction{G}{\maxLB}$ generally must be preserved exactly.
This gives us a lower bound on the best achievable size guarantee of $\Omega(n^{2/3} |K|^{2/3})$, i.e., $\Omega(n^2)$ or $\Omega(m)$ for $|K| \in \Omega(n^2)$~\cite{coppersmith2006}.
Hence, there is no way to guarantee further sparsification of $\weightRestriction{G}{\maxLB}$.
Yet, by calling \GR on this restricted graph $\weightRestriction{G}{\maxLB}$, but with unaltered terminal pairs $K$, lengths $\ell$, and distance demands $\demand$, we have a simple way to generate a feasible spanner that does obtain stronger weight and size guarantees for several interesting instance classes.

\begin{theorem}\label{thm:2PhaseGreedy}
    \AG is an %
    $m$-approximation for all (un)directed \generic \freeform\ spanner problems.
\end{theorem}
\begin{proof}
    $\lowG$ is a feasible spanner by construction.
    Thus, for all $(u,v) \in K$, the shortest $uv$-path in $\lowG$ has length at most $\demand(u,v)$ and
    \Cref{alg:AdaptedGreedy} returns a feasible spanner.
    The weight guarantee follows from $w(H)\leq w(E[\maxLB]) \leq |E[\maxLB]| \cdot \maxLB \leq m \cdot \opt$.

    The running time is dominated by $|K|\log m$ and $2|K|$ many shortest path computations in phase 1 and 2, respectively.
    Thus, it is in $\mathcal{O}(\myAGTime)$.
\end{proof}

\Cref{thm:2PhaseGreedy} is the first approximation algorithm for the \generic \freeform spanner problem that does not require polynomially bounded integer lengths and, e.g., is applicable to instances with Euclidean length.
Further, it works for both directed and undirected graphs.
As previously mentioned, using known bounds for distance preservers, i.e., Theorem 1.1(2) in~\cite{coppersmith2006}, we can construct instances for which our analysis of ratio $m$ is tight.
The running time matches \GR up to a factor of $\mathcal{O}(\log m)$.
Note that sparsifiying $G$ prior to phase 1 generally leads to the loss of any weight guarantee.

In many applications, edge weight and length are strongly anti-correlated; e.g., in network design, low-delay connections are typically more expensive. Here, \GR tends to struggle, as it hastily selects all of the shortest, and thus most expensive, connections. Hence, it cannot yield any weight guarantees.
Conversely, \AG and its analysis can benefit from such scenarios:
Especially if distance demands are somewhat similar to the original distances,
it is in practice often reasonable to assume that one has to select at least one short but comparably expensive edge into the spanner. This yields a large $\maxLB$, and thus there may be only few \emph{high-weight} edges $\highE\coloneqq E \setminus \lowE$.
In such settings, we can use \GR's absolute size guarantees to deduce stronger weight guarantees. 
Recall that for \generic multiplicative $\alpha$-spanners, by \Cref{lem:metricTerminalPairs} and~\cite{peleg1989}, we only have to enforce distance demands for all metric edges $M \subseteq E$.

\begin{lemma}\label{lem:AGSG}
    Consider the undirected \generic multiplicative $\alpha$-spanner problem. 
    \AG (for $\alpha=2k-1$, $k \in \mathbb{N}_{\geq 1}$) yields a spanner of size and of weight ratio $\mathcal{O}(n^{1+1/k} + \sqrt{n}|\highE|)$. In the geometric setting (for $\alpha>1$), this improves to $\mathcal{O}(n+\sqrt{n}|\highE|)$.
\end{lemma}
\begin{proof}
    For each $\{u,v\} \in M$, the shortest $uv$-path in $G$ by definition consists only of the edge itself.
    Crucially, all $M \cap \lowE$ still exist in the weight-restricted graph $\lowG$.
    Thus, covering these edges in phase 2 is analogous to using \GR on $\lowG$ with terminal pairs $M \cap \lowE$.
    By \GR's size guarantees, this introduces at most $\mathcal{O}(n^{1+1/k})$ edges, in general, or $\mathcal{O}(n)$ edges in the geometric setting.
    In contrast, for each edge in $M \cap \highE$, we may have to add a path consisting of up to $n-1$ low-weight edges. However, recall that we compute unique shortest paths achieved by tie-breaking. Thus, by~\cite[Corollary 7.8]{coppersmith2006}, originally formulated in the context of pairwise distance preservers, the overall number of added edges can be upper bounded by $\mathcal{O}(n + \sqrt{n} |\highE|)$.
\end{proof}

\begin{corollary}\label{cor:2PhaseGreedyExpensive}
    If $|\highE| \in o(m/ \sqrt{n})$, better approximation guarantees than ratio $m$ become possible.
    In particular, the weight ratio and size are in $\mathcal{O}(n^{1+1/k})$ in general for $|\highE| \in \mathcal{O}(\sqrt{n}\cdot n^{1/k})$, and in $\mathcal{O}(n)$ in the geometric setting for $|\highE| \in \mathcal{O}(\sqrt{n})$.
\end{corollary}

Note that even for $|\highE| \in \mathcal{O}(1)$, in fact even $|\highE| = 1$, \GR has no weight guarantee on \generic instances.
Next, we show that, despite it working on a restricted graph to allow approximating decoupled instances, \AG retains \GR's size and lightness guarantees in the established contexts. %

\begin{lemma}\label{lem:2PhaseGreedyRetention}
    \AG has the same size and lightness guarantees as \GR for undirected \uniform multiplicative $\alpha$-spanners, for all $\alpha \geq 1$.
\end{lemma}
\begin{proof}
    Let \GR and \AG run in parallel. Recall that $K=M$ suffices. Both process edges in the same order, starting with $\lowE$.
    For all $e \in  \lowE$, they behave identically. 
    \GR always includes an MST~\cite{althofer1993} and note that $\maxLB \geq \wMST$.
    Hence, after covering all $e \in  \lowE$, \AG contains an MST, too,
    and the weight and thus length of every $e \in \highE$ is larger than $\wMST$.
    Thus, for any $e \in \highE$, a sufficiently short path through the MST already exists and \AG does not add any more edges. %
    For \uniform $\alpha$-spanners, \GR and \AG yield the same spanner.
\end{proof}

\begin{corollary}\label{cor:2PhaseGreedyRetention}
    \AG yields undirected \uniform $\alpha$-spanners of lightness (and thus weight ratio) $\mathcal{O}(n/k)$ and size $\mathcal{O}(n^{1+1/k}$), where $\alpha=2k-1$, $k \in \mathbb{N}_{\geq 1}$.
    In the geometric setting, these improve to $\mathcal{O}(\log n)$ and $\mathcal{O}(n)$, respectively, for all $\alpha>1$.
    For $\alpha \geq \log n$ and $\alpha \geq (\log n)^2$, lightness even is in $\mathcal{O}(\log n)$ and $\mathcal{O}(1)$, respectively.
    Violating distance demands by a factor of $(1+\varepsilon)$, $\varepsilon>0$, we yield $(1+\varepsilon) (2k-1)$-spanners of lightness $\mathcal{O}(n^{1/k} (1+k/(\varepsilon^{1+1/k} \log k)))$.
\end{corollary}
These four statements follow from combining \Cref{lem:2PhaseGreedyRetention} with~\cite{althofer1993}, \cite{chandra1992}, \cite{chandra1992}, and~\cite{elkin2015light}, respectively.
\Cref{cor:2PhaseGreedyRetention,cor:2PhaseGreedyExpensive} 
also hold for undirected \generic linear $(\alpha, \beta)$-spanners, for all $\beta\geq 0$.

For undirected basic additive $+\beta$-spanners and $\beta=2$, %
\GR guarantees size in $\mathcal{O}(n^{3/2})$~\cite{knudsen2014}.
The proof can easily be generalized to accommodate \AG by adjusting some minor details, even in the unit-length setting and for any constant $\beta \geq 2$.
As the resulting proof works analogous to~\cite{knudsen2014}, we refer to~\Cref{appendix:AdaptedGreedy}.
The key observation is that, to achieve size $\mathcal{O}(n^{3/2})$, the $uv$-paths added between $\{u,v\} \in K$ during construction, do not necessarily have to be shortest paths. %
For any constant $\beta \geq 2$, paths of length at most $\demand(u,v)=\dist{G}{\ell}(u,v) + \beta$ suffice. %

\begin{corollary}\label{cor:AdaptedGreedy+beta}
    \AG yields undirected unit-length additive $+\beta$-spanners of size and weight ratio $\mathcal{O}(n^{3/2})$, for any constant $\beta \geq 2$.
\end{corollary}
Hence, \AG not only matches the best possible~\cite{woodruff2006} size guarantee for $\beta=2$, 
but derives new weight guarantees in the unit-length setting, even for all constant $\beta \geq 2$.%

\section{Randomized Rounding}\label{sec:RandomizedRounding}
We now describe our $\ln(\myRRratio) \in \mathcal{O}(n \log n)$-approximation algorithm for the (un)directed \generic \freeform\ spanner problem, with maximum distance demand $\maxDemand$ and $\maxNumCuts=(\maxDemand + 2)^{n-2}$. %
As the other state-of-the-art algorithms for decoupled spanner problems~\cite{chimani2014network,grigorescu2023approximation} %
we assume integer lengths (and, thus, w.l.o.g.\ integer distance demands) and polynomially bounded distance demands.
The running time is polynomial in $n$, and the approximation ratio of $\ln(\myRRratio)$ 
is in $\mathcal{O}(n \log n)$. 
For simplicity, we assume a directed setting. Undirected settings work analogous, by bi-directing the graph.

We use multicommodity flow in layered graphs.
Dodis and Khanna~\cite{dodis1999} apply a similar technique to unit-length multiplicative $\alpha$-spanners.
Our generalization can handle both non-unit lengths and \freeform\ distance demands.
A key challenge lies in the non-unit lengths.
In contrast to~\cite{dodis1999}  (see \Cref{sec:preliminaries}), we do not subdivide edges and we do not use a unit-length spanner algorithm as a subroutine.
Instead, we encode lengths directly in the graph extension.
Ultimately, our construction allows us to effortlessly enforce distance demands of any form for any subset of node pairs.
We establish the following graph extension; for easier distinction, we use the term \emph{arc} to refer to the directed edges in the extension.

\begin{definition}[$\maxDemand$-extension]\label{def:lambdaExtension}
    The \emph{$\maxDemand$-extension} of a directed graph $G=(V,E)$ with integer lengths and distance demands
    is a layered directed graph $\extension{G} = (\extension{V}, \extension{E})$ with:
    \begin{itemize}
        \item $\maxDemand + 1$ layers of nodes $V_0^{\maxDemand}, V_1^{\maxDemand}, \dots, V_{\maxDemand}^{\maxDemand}$ where each $V_i^{\maxDemand}$ is a copy of $V$. For $q \in V$, $q_i$ denotes the copy of $q$ in $V_i^{\maxDemand}$. %
        \item For each edge $(s,t)= e \in E$, %
        there are arcs $(s_i, t_{i+ \ell(e)})$, for all $0 \leq i \leq \maxDemand - \ell(e)$. Additionally, for all $q \in V$, there are \emph{self-arcs} $(q_i, q_{i+1})$, for all $0 \leq i < \maxDemand$. 
    \end{itemize}
\end{definition}

Clearly, $|\extension{V}| = n(\maxDemand + 1)$ and $|\extension{E}| \leq (n+m) \maxDemand$. \Cref{exa:DeltaExtension} illustrates our construction.

\begin{example}\label{exa:DeltaExtension}
    \Cref{fig:myExample} shows a simple example.
We consider a directed \generic \freeform\ spanner problem instance with graph $G=(V,E)$, $V=\{a,b,c\}$, $E=\{(a,b), (a,c), (c,b)\}$, terminal pairs $K=E$, weights $w(a,b)=5$, $w(a,c)=w(c,b)=1$, lengths $\ell(a,b)=\ell(c,b)=1$, $\ell(a,c)=2$, and distance demands $\demand(a,b)=3$, $\demand(a,c)=\demand(c,b)=2$. Thus, $\maxDemand=3$. 
Clearly, the optimum solution to this instance is $H=(V,E')$ with $E'=\{ (a,c), (c,b) \}$ of total weight $2$.\hfill $\diamond$
\end{example}

\begin{figure}
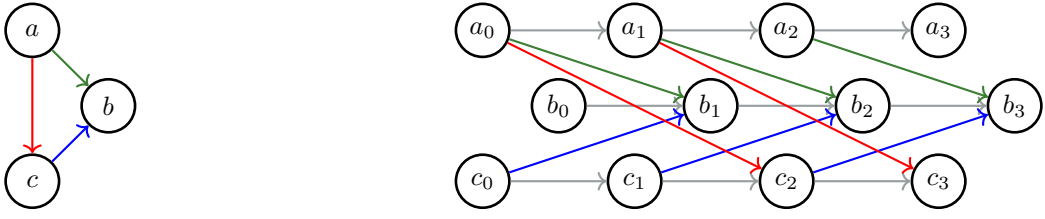

    \centering
        \begin{subfigure}{0.30\textwidth}
        \centering
        \includeTikzGraph{myTriangle}
    \end{subfigure}
    \hfill
    \begin{subfigure}{0.65\textwidth}
        \centering
        \includeTikzGraph{myLayeredTriangle}
    \end{subfigure}
    \caption{The graph $G$ (left) and its $3$-extension (right), as described in \Cref{exa:DeltaExtension}.}
    \label{fig:myExample}
\end{figure}

We state the spanner problem as a standard multicommodity flow LP in the $\maxDemand$-extension of $G$, where we ship one commodity from $u_0$ to $v_{\demand(u,v)}$, for each $(u,v) \in K$.
We then round the fractional solution to obtain a subgraph $H=(V,E')$ and prove that $H$ forms a feasible spanner w.h.p.\ via a simple cut argument.
The full algorithm is shown in \Cref{alg:randomizedRounding}.
Let us first discuss the formulation \MCF below:
There are \emph{flow variables}~\eqref{lp:ours:flowVars} and \emph{edge variables}~\eqref{lp:ours:arcVars}.
We use standard flow conservation constraints~\eqref{lp:ours:flowCons}, where 
the indicator function $\mathbbm{1}_\varphi$ is $1$ if the Boolean expression $\varphi$ is true, and $0$ otherwise.
Constraints~\eqref{lp:ours:edgeCons} ensure that if any non-self-arc in $\extension{E}$ carries flow, its corresponding edge variable must be active. The objective function~\eqref{lp:ours:obj} considers the edge variables to minimize the weight of the spanner.
In total, \MCF contains $\mathcal{O}(\myNumVars)$ variables, $\mathcal{O}(\myNumCons)$ constraints, and the largest coefficient stems from the input weights. It can directly be solved using any standard (polynomial-time) LP solver.
\label{lp:ours}
\begin{align}
\hspace*{-5mm}\MCF\qquad  & \text{min} \sum_{e \in E} w(e) x_{e} & \label{lp:ours:obj}\\
\text{s.t.\ } &        \sum_{i=0}^{\maxDemand - \ell(e)}  f^{uv}_{(s_i,t_{i+\ell(e)})}  \leq x_e
                            & \forall (u,v) \in K,\forall e=(s,t) \in E \label{lp:ours:edgeCons}\\
&             \sum_{(q,\cdot) \in \extension{E}} f^{uv}_{(q,\cdot)} 
                - \sum_{(\cdot,q) \in \extension{E}} f^{uv}_{(\cdot,q)} = \mathbbm{1}_{q = u_0} - \mathrlap{\mathbbm{1}_{q = v_{\demand(u,v)}}} & \forall (u,v) \in K, \forall q \in \extension{V} \label{lp:ours:flowCons}\\
&              0 \leq f^{uv}_{(q,r)} \leq 1 & \forall (u,v) \in K, \forall (q,r) \in \extension{E} \label{lp:ours:flowVars} \\
&              0 \leq x_{e} \leq 1 & \forall e \in E \label{lp:ours:arcVars}
\end{align}%

We now derive a subgraph $H=(V,E')$ from an optimal solution to \MCF.
Let $x_e^*$ denote an optimal (fractional) solution value for variable $x_e$.
Let $\gamma$ be a value defined later. %
For each edge $e\in E$, we include $e$ in $E'$ with probability $\min\{1, \gamma x_e^*\}$. %
The expected total weight of $E'$ is then at most a factor of $\gamma$ away from the optimum spanner.

\begin{algorithm}[bt]
\caption{\RR.}\label{alg:randomizedRounding}
\begin{algorithmic}[1]
\Require graph $G=(V,E)$, terminal pairs $K$, weights $w$, lengths $\ell$, distance demands $\demand$
\State solve \MCF $\rightarrow$ optimal fractional solution values $x_e^*$ for all $e \in E$ 
\State $\gamma \leftarrow \ln(\myRRratio)$ with $\maxNumCuts=(\maxDemand + 2)^{n-2}$
\State $E'\leftarrow\emptyset$
\For{$e \in E$}
    \State add $e$ to $E'$ with probability $\min\{1, \gamma x_e^*\}$
\EndFor
\State \textbf{return} $H=(V,E')$
\end{algorithmic}
\end{algorithm}

We show that $H$ forms a feasible spanner with high probability by a suitable choice for~$\gamma$.
To this end, we use a simple cut argument in the $\maxDemand$-extension $\extension{H}=(\extension{V}, \extension{E'})$ of $H$.
A $u_0v_{\demand(u,v)}$-cut $(A{:}B)$ partitions the nodes $\extension{V}$ into disjoint subsets $A$ and $B$ with $u_0 \in A$ and $v_{\demand(u,v)} \in B$.
A cut $(A{:}B)$ is \emph{satisfied} if at least one arc in $\extension{E'}$ goes from $A$ to $B$, see
\Cref{fig:myAscendingCutExample}. %
\begin{figure}
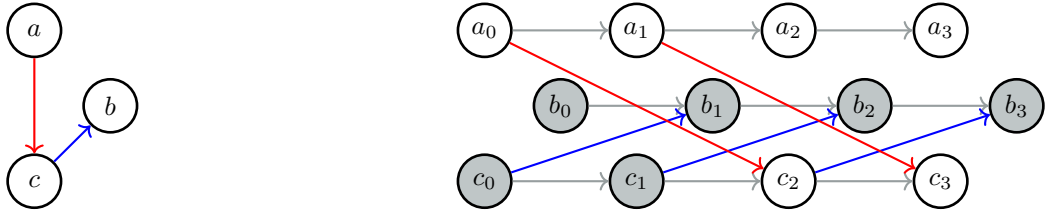

    \centering
        \begin{subfigure}{0.30\textwidth}
        \centering
        \includeTikzGraph{myTriangleSolution}
    \end{subfigure}
    \hfill
    \begin{subfigure}{0.65\textwidth}
        \centering
        \includeTikzGraph{myAscendingCutExampleInH}
    \end{subfigure}
    \caption{The optimum solution $H \subset G$ of \Cref{exa:DeltaExtension} (left) and an $a_0b_3$-cut $(A{:}B)$ ($A$ is white and $B$ is gray) in its $3$-extension satisfied by arc $(c_2, b_3)$ (right). The cut is ascending, induced by the node labeling $\myLabeling(a)=0$, $\myLabeling(b)=\demand(a,b)+1=4$, and $\myLabeling(c)=2$.}
    \label{fig:myAscendingCutExample}
\end{figure}
Recall that $H$ is feasible if and only if $d_{H}^{\ell}(u,v) \leq \demand(u,v)$, for all $(u,v) \in K$.%

\begin{lemma}\label{lem:lambdaCutEquivalence}
    Let $H=(V,E') \subseteq G$ with $\maxDemand$-extension $\extension{H}$ and any $(u,v) \in K$. Then %
    $$ \textsf{\emph{\textbf{(1)\ }}} d_{H}^{\ell}(u,v) \leq \demand(u,v) \Leftrightarrow \textsf{\emph{\textbf{(2)\ }}} \text{\emph{every $u_0v_{\demand(u,v)}$-cut in $\extension{H}$ is satisfied.}} $$
\end{lemma}
\begin{proof}
    Assume (1) %
    and let $P \subseteq E'$ be a shortest $uv$-path in $H=(V,E')$.
    Then $\ell(P) \leq \demand(u,v)$.
    We can construct a $u_0 v_{\demand(u,v)}$-path $\extension{P} \subseteq \extension{E'}$ in $\extension{H}$ using the arcs corresponding to edges in $P$ (and self-arcs), and (2) holds.
    Conversely, if (2), there is a path $\extension{Q}$ from $u_0$ to $v_{\demand(u,v)}$ in $\extension{H}$. $\extension{Q}$ induces a $uv$-path of length at most $\demand(u,v)$ in $H$ (removing self-arcs). %
\end{proof}

We argue that only a certain kind of cuts may ever be non-satisfied in $\extension{H}$.

\begin{definition}[Ascending $u_0v_{\demand(u,v)}$-cut]\label{def:ascendingCut}
    Let $H=(V,E')$ with $\maxDemand$-extension $\extension{H}$ and $u,v \in V$.
    A $u_0v_{\demand(u,v)}$-cut $(A{:}B)$ in $\extension{H}$ is \emph{ascending} if it is \emph{induced} by a node labeling $\myLabeling \colon V \rightarrow \{0,1,\dots, \demand(u,v)+1\}$, 
    i.e., for every $q \in V$ we have that $q_i \in B$ for all $0 \leq i < \myLabeling(q)$, and $q_i \in A$ for all $\myLabeling(q) \leq i \leq \demand(u,v)+1$.
\end{definition}

Necessarily, $\myLabeling(u) = 0$ and $\myLabeling(v) = \demand(u,v) +1$.
Observe that \Cref{fig:myAscendingCutExample} indeed shows an ascending cut.
A $u_0v_{\demand(u,v)}$-cut $(A{:}B)$ is non-ascending if and only if there is a $q \in V$  %
with $q_i \in A$ and $q_{i+1} \in B$, for some $0 \leq i \leq \demand(u,v)$. Such a cut is trivially satisfied by self-arc $(q_i, q_{i+1})$.
Hence, we only have to consider ascending cuts.
For $(u,v) \in K$, let $\mathcal{C}_{uv}$ denote the set of all ascending $u_0v_{\demand(u,v)}$-cuts in~$\extension{H}$.
By construction via the node labeling $\myLabeling$, we get $|\mathcal{C}_{uv}| = (\demand(u,v)+2)^{n-2}$.
Define $\maxNumCuts \coloneqq \max_{(u,v) \in K} |\mathcal{C}_{uv}| \leq (\maxDemand + 2)^{n-2}$ as the maximum number of ascending $u_0v_{\demand(u,v)}$-cuts for any $(u,v) \in K$.
Next, we bound the probability that an ascending cut is non-satisfied in the $\maxDemand$-extension of our rounded subgraph.

\begin{observation}\label{lem:probability}
        Let $X_1, \dots, X_k$ be independent, binary, random variables with $\Prob(X_i) = \min\{1, p_i\}$ and $\sum_{i=1}^{k} p_i\geq S$. Then $\Prob[\bigwedge_{i=1}^{k} (X_i=0)] = \prod_{i=1}^{k}(1-p_i) \leq e^{-( \sum_{i=1}^{k}p_i)} \leq e^{-S}$. 
\end{observation}

\begin{lemma}\label{lem:cutSatisfyProbability}
    For $(u,v) \in K$, a given ascending $u_0v_{\demand(u,v)}$-cut $(A{:}B)$ is non-satisfied in $\extension{H}$ with probability at most $e^{-\gamma}$.
\end{lemma}
\begin{proof}
    Let $\extension{E'}_c\subseteq \extension{E'}$ be the set of arcs %
    that go from $A$ to $B$ in $\extension{H}$. %
    Let $E'_c \subseteq E'$ be their corresponding edges in $H$.
    Since one unit of flow must cross from $A$ to $B$ in $\extension{H}$, we get
    $1 \leq \sum_{a \in \extension{E'}_c} {f^{uv}_a}^* \leq \sum_{e \in E'_c} x_e^* $, where ${f^{uv}_a}^*$ are the optimal solution values for $f^{uv}_a$, respectively.
    Each $e \in E'_c$ is chosen with probability $\min\{1, \gamma x_e^*\}$ and $\sum_{e \in E_c} \gamma x_e^* \geq \gamma$. 
    Thus, by 
    \Cref{lem:probability},
    we get $\Prob[(A{:}B)\text{ is not satisfied}] =\Prob[E_c \cap E' = \emptyset]\leq e^{-\gamma}$.
\end{proof}

We are now ready to define $\gamma \coloneqq \ln(n\maxNumCuts|K|)$ and get $e^{-\gamma} = \frac{1}{n\maxNumCuts|K|}$.
By \Cref{lem:lambdaCutEquivalence}, $H$ is infeasible if and only if, for any $(u,v) \in K$, there is an ascending $u_0v_{\demand(u,v)}$-cut that is not satisfied.
By union bound, we can estimate this probability as
\begin{equation}%
\label{eq:probFeas}
    \sum_{(u,v) \in K} \sum_{(A{:}B) \in \mathcal{C}_{uv}} e^{-\gamma} \leq  \maxNumCuts |K| \cdot  e^{-\gamma} = \frac{\maxNumCuts |K|}{\myRRratio} = \frac{1}{n}. 
\end{equation}
Hence, our rounded solution $H=(V,E')$ forms a feasible spanner with probability at least $1-\frac{1}{n}$.
The total running time of \Cref{alg:randomizedRounding} is in $\mathcal{O}(\myLP)$; in particular, it is polynomial in $n$ if $\maxDemand \in \mathcal{O}(\poly(n))$.
Using $\maxDemand \leq n L$ and $n \leq m$ yields an upper bound of $\mathcal{O}(\myLPUB)$.
Thus

\begin{theorem}\label{thm:RandomizedRounding}
    For (un)directed \generic \freeform\ spanner problems with integer lengths and polynomially bounded distance demands, \Cref{alg:randomizedRounding} returns a feasible spanner w.h.p., with expected approximation ratio $\gamma = \ln(n \maxNumCuts|K| ) \leq \ln(n^3 (\maxDemand + 2)^{n-2}) \in \mathcal{O}(n \log \maxDemand)\subseteq \mathcal{O}(n \log n)$.
\end{theorem}

If the instance is undirected, we construct the $\maxDemand$-extension of its bi-directed counterpart but only have undirected edge variables in \eqref{lp:ours:arcVars}; constraints \eqref{lp:ours:edgeCons} are constructed twice, once per possible direction of $e$.
The remaining constructions and proofs then work identically. %

\subparagraph{Comparison to Grigorescu et al.~\cite{grigorescu2023approximation}.}
It is worthwhile to compare our algorithm to the currently theoretically strongest known approach~\cite{grigorescu2023approximation} for this setting. Amongst other steps, the latter requires the computation of up to $\GrigNumTrees$ minimum-weight distance-preserving junction trees. Each such tree requires solving an LP with orders of magnitude more variables and constraints; namely, each such computation requires $\Omega(\GrigLP)$ time.
Clearly, these computations alone, even disregarding all other necessary subroutines (solving Path-LPs, etc.), drastically exceed our time guarantee.
Also, neglecting all additional (expected) weight for covering the remaining terminal pairs, the weight introduced by their constrained shortest path trees is lower bounded by $12 \cdot n^{4/5} \ln n \cdot \opt$. 
Compare this to our expected ratio of $\gamma = \ln(n \maxNumCuts |K|) \leq \ln(n^3 (\maxDemand + 2)^{n-2})$.
For, e.g., instances with 
$\maxDemand \leq n$, $\sqrt{n}$, $\log n$, or $10$, our algorithm expects a superior approximation ratio for all $n \leq 248\, 000$, $10^6$, $10^{9}$, and $10^{10}$, respectively. %
Clearly, this encompasses many practically relevant instances.

\subparagraph{Constant-bounded instances.}
Consider any instance in the scenario of~\Cref{thm:RandomizedRounding}; we say it is \emph{constant-bounded} if the maximum distance demand $\maxDemand$ is bounded by some constant.
\begin{corollary}\label{cor:RandomizedRoundingConstantDelta}
    \Cref{alg:randomizedRounding} is an $\mathcal{O}(n)$-approximation for constant-bounded instances.
\end{corollary}

Currently, among constant-bounded spanner problems,
only undirected unit-length multiplicative $2$-spanners (where thus $\delta(u,v)=2$ for all $\{u,v\}\in K=E$) are known to allow approximation ratios in $\mathcal{O}(\log n)$~\cite{kortsarz1994, kortsarz2001}.
We extend this result to all constant-bounded \generic \freeform spanner problems with constant maximum (out)degree $\maxDeg$, even on directed graphs.
Recall that under these restriction, even the undirected basic multiplicative $\alpha$-spanner problem is \textbf{NP}-hard~\cite{gomez2023}. Further, without restricting the (out)degrees, there cannot be a polylogarithmic approximation ratio unless \textbf{NP}$\subseteq$\textbf{BPTIME}$(n^{\text{polylog}(n)})$~\cite{dinitz2015}.

For $(u,v) \in K$ and $G=(V,E)$, let subgraph $\myRestriction{G}=(\myRestriction{V}, \myRestriction{E})$, with $\myRestriction{V}\subseteq V$ and $\myRestriction{E}\subseteq E$, be induced by all $uv$-paths $P$ in $G$ with $\ell(P) \leq \demand(u,v)$.
Then $\myRestriction{G}$ contains all nodes and edges that affect whether terminal pair $(u,v)$ satisfies its distance demand.

\begin{lemma}\label{lem:restrictedGCuts}
    Every ascending $u_0v_{\demand(u,v)}$-cut in $\mathcal{C}_{uv}$ in $\extension{H}$ is satisfied if and only if every ascending $u_0v_{\demand(u,v)}$-cut in $\mathcal{C}_{uv}'$ in $\extension{\myRestriction{H}}$ is satisfied.
\end{lemma}
\begin{proof}
    $\mathcal{C}_{uv} \subseteq \mathcal{C}_{uv}'$ gives necessity.
    For sufficiency, note that no satisfied cut in $\mathcal{C}_{uv}'$ induced by a labeling $\myLabeling$ of nodes $\myRestriction{V}$ can be dissatisfied by expanding $\myLabeling$ to label further nodes. %
\end{proof}

By \Cref{lem:restrictedGCuts}, we only have to choose $\gamma$ such that all ascending $u_0v_{\demand(u,v)}$-cuts in $\extension{\myRestriction{H}}$ are satisfied w.h.p.
Note that $n_{uv} \coloneqq |\myRestriction{V}| \leq \maxDeg^{\demand(u,v)}$, for every $(u,v) \in K$.
Hence, the number of relevant ascending cuts $|\mathcal{C}_{uv}'|$ in $\extension{\myRestriction{H}}$ is bounded by $|\mathcal{C}_{uv}'| \leq (\demand(u,v)+2)^{n_{uv}-2}$.
Let $\maxNumCuts' \coloneqq \max_{(u,v) \in K} |\mathcal{C}_{uv}'|$ the maximum number of relevant cuts for any $(u,v) \in K$.
Then, $\gamma = \ln(n \maxNumCuts' |K|) \in \mathcal{O}(\log n + \maxDeg^{\maxDemand} \log \maxDemand)$ suffices to guarantee 
feasible solutions w.h.p. 
Hence, the expected approximation ratio improves with decreasing $\maxDeg^{\maxDemand}$ and, 
e.g., beats the best known ratio for unit-length multiplicative $3$-spanners of $\widetilde{\mathcal{O}}(n^{1/2})$~\cite{dinitz2011directed} if $\Delta \in o(n^{1/6})$.
Crucially, for constant $\maxDemand$ and constant $\maxDeg$, choosing $\gamma = \ln(n \maxNumCuts'|K|) \in \mathcal{O}(\log(n |K|)) \subseteq \mathcal{O}(\log n)$ suffices.

\begin{corollary}\label{cor:RandomizedRoundingConstDR}
     \Cref{alg:randomizedRounding} is an $\mathcal{O}(\log n)$-approximation for constant-bounded instances with constant maximum (out)degree.
\end{corollary}

Choosing $\gamma = \ln(2 \maxNumCuts' |K|)$ still yields feasible solutions with probability at least 
$1- \frac{1}{2} = \frac{1}{2}$ (see \eqref{eq:probFeas}), but no longer w.h.p.
Then, for constant-bounded instances with constant maximum (out)degree, if the number of terminal pairs $|K|$ is sublinear in $n$ or constant, the so-modified \Cref{alg:randomizedRounding} achieves expected weight ratios in $o(\log n)$ or $\mathcal{O}(1)$, respectively.

\section{Conclusion}
We provided two simple approximations for \generic \freeform spanner problems (cf.~\Cref{tab:myResults,tab:AGResults}). %
Our \AG yields the first unconditional $m$-approximation for this problem.
At the same time, it retains all weight and size guarantees \GR is appreciated for. %
For integer lengths and polynomially bounded distance demands, our \RR has ratio $\mathcal{O}(n \log n)$, closely matching the best known ratio under these conditions~\cite{grigorescu2023approximation}. %
In contrast to the latter very complex %
approach, we simply round a solution of a standard multicommodity flow LP in a graph extension.
On bounded-degree graphs, this furthermore provides the first $\mathcal{O}(\log n)$-ratio for (undirected or directed) spanner problems with constant-bounded distance demands (beyond undirected unit-length $2$-spanners).

Future work could include practical studies.
Such evaluations would have to focus on algorithms tuned for simplicity and implementability, as
our proposed algorithms.
Also, comparisons against the \enquote{defacto state-of-the-art} of forcing \generic instances into a \uniform format would be interesting.
Implementations of algorithms like~\cite{grigorescu2023approximation} seem unrealistic, due to their high complexity and reliance on never implemented subroutines.

\clearpage

\bibliography{bibliography.bib}

\clearpage
\appendix

\section*{APPENDIX}

\section{Error in~\cite{dodis1999}}\label{appendix:Dodis}
In~\cite{dodis1999}, the authors give a (sound) $\mathcal{O}(n \log \alpha)$-approximation for directed unit-length multiplicative $\alpha$-spanners.
In the end, they briefly (and unfortunately wrongly) suggest how to generalize this algorithm to the \generic setting with polynomially bounded lengths.
They subdivide arcs and use the unit-length algorithm to solve the transformed instance.
However, we show that this approach does not yield feasible solutions whenever there is an $e \in E$ with $\ell(e)>\alpha$. \Cref{fig:DK_counterexample} showcases this.
The intuitive downfall is that the unit-length algorithm assumes distance demands for all arcs.
This, however, is not given by their transformation.

\begin{figure}[h]
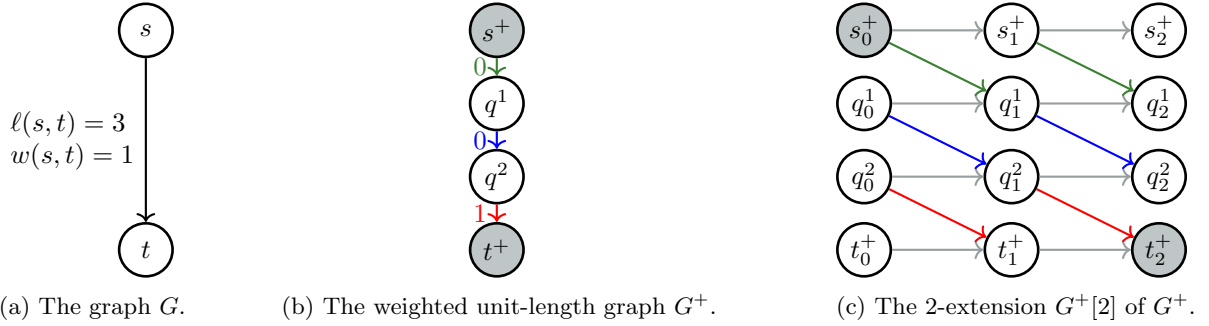

    \centering
        \begin{subfigure}{0.2\textwidth}
        \centering
        \includeTikzGraph{GraphG}
        \subcaption{The graph $G$.}
        \label{fig:DK_counterexampleA}
    \end{subfigure}
    \hfill
    \begin{subfigure}{0.37\textwidth}
        \centering
        \includeTikzGraph{GraphGPlus}
        \subcaption{The weighted unit-length graph $G^+$.}
        \label{fig:DK_counterexampleB}
    \end{subfigure}
    \hfill
    \begin{subfigure}{0.37\textwidth}
        \centering
        \includeTikzGraph{LayeredGraphGPlus}
        \subcaption{The $2$-extension $G^+[2]$ of $G^+$.}
        \label{fig:DK_counterexampleC}
    \end{subfigure}
    \caption{A simple counterexample. The model is infeasible as no flow can be send from $s^+_0$ to $t^+_2$.}
    \label{fig:DK_counterexample}
\end{figure}

Consider a \generic instance $\mathcal{I}$ with graph $G=(V,E)$, weights $w(e)\geq 0$, and polynomially bound integer lengths $\ell(e) > 0$ for all $e \in E$.
First, the algorithm transforms $\mathcal{I}$ into a unit-length instance $\mathcal{I}^+$ with graph $G^+=(V^+, E^+)$, weights $w^+(e)\geq 0$, and lengths $\ell^+(e) = 1$ for all $e \in E^+$.
For every $e=(s,t) \in E$, $G^+$ contains a path of $\ell(e)$-many arcs $(s^+=q^0,q^{1}), (q^{1}, q^{2}),\dots, (q^{\ell(e)-1}, q^{\ell(e)}=t_+)$, with $s^+$ and $t^+$ corresponding to $s$ and $t$.
It sets $w^+(q^i, q^{i+1}) = 0$ for all $0 \leq i < \ell(e)-1$, $w^+(q^{\ell(e)-1}, t^+)=w(e)$, and $\ell^+(e)=1$ for all $e \in E^+$.
Distance demands are enforced for all node pairs corresponding to nodes that are adjacent in $G$.
That is $K^+=\{ (s^+,t^+) \subseteq V^+ \times V^+ \mid (s,t) \in E\}$.

\Cref{fig:DK_counterexampleA} shows an instance of the directed \generic multiplicative $2$-spanner problem with $G=(V,E)$, $V=\{s,t\}$, $E=\{(s,t)\}$, $w(s,t)=1$, and $\ell(s,t)=3$.
Clearly, $H=(V,E'=\{(s,t)\})$ forms a feasible solution. 
\Cref{fig:DK_counterexampleB} shows the directed unit-length graph $G^+=(V^+, E^+)$ constructed from $G$, with $K^+=\{(s^+, t^+)\}$ marked grey.

Next, the algorithm solves the newly created instance $\mathcal{I}^+$ with their previously proposed unit-length $\alpha$-spanner approximation.
It constructs an $(\alpha+1)$-layered directed graph $G^+[\alpha]=(V^+[\alpha],E^+[\alpha])$ (called the $\alpha$-extension) derived from $G^+$.
$G^+[\alpha]$ has $\alpha + 1$ layers of node sets $V_0^+, \dots, V_\alpha^+$ where each $V_i^+$ is a copy of $V^+$.
For $q \in V^+$, $q_i$ denotes the copy of $q$ in $V_i^+$.
For all $(x,y) \in E^+$, there are arcs $(x_i, y_{i+1})$, for all $0 \leq i < \alpha$.
Further, there are \emph{self-arcs} $(q_i, q_{i+1})$ for all $q \in V^+$ and $0 \leq i < \alpha$.
By traversing any arc, one ascends exactly one layer.
The algorithm solves the LP relaxation of a multicommodity flow ILP.
For every $(u, v) \in K^+$, the ILP sends one unit of flow through $G^+[\alpha]$ from $u_0$ to $v_\alpha$.
The spanner is then obtained via randomized rounding of the fractional solution.
\Cref{fig:DK_counterexampleC} shows the $2$-extension $G^+[2]$ of $G^+$ in our example.
The algorithm must send one unit of flow from $s^+_0$ to $t^+_2$ through $G^+[2]$.
However, due to the subdivision of $(s,t)$, no such flow exists, as there is no directed path from $s^+_0$ to $t^+_2$ in $G^+[2]$.
Consequently, the algorithm cannot yield a feasible solution.
\Cref{obs:flow_length} formalizes this.

\begin{observation}\label{obs:flow_length}
    The generalized algorithm for \generic multiplicative $\alpha$-spanners~\cite[Section 5.2]{dodis1999} is infeasible if any $(s,t) \in E$ has $\ell(s,t)>\alpha$.
\end{observation}

\section{Proof of \Cref{cor:AdaptedGreedy+beta}: Augmented Greedy on unit-length $+\beta$-spanners}\label{appendix:AdaptedGreedy}
We show that \AG yields undirected unit-length additive $+\beta$-spanners with size and weight ratio $\mathcal{O}(n^{3/2})$, for all constant integer $\beta \geq 2$.
For $\beta=2$, the size guarantee matches the lower bound~\cite{woodruff2006}.
Knudsen~\cite{knudsen2014} shows that \GR\footnote{The proof even works regardless of the order in which terminal pairs are processed.} yields basic additive $+2$-spanners with size $\mathcal{O}(n^{3/2})$.

We generalize the proof in two ways:
Firstly, we show that \GR yields additive $+\beta$-spanners of size $\mathcal{O}(n^{3/2})$, for all constant integer $\beta \geq 2$.
Secondly, we show that the $uv$-paths added for terminal pairs $\{u,v\} \in K$ during construction, do not necessarily have to be shortest paths in order to achieve the size guarantee of $\mathcal{O}(n^{3/2})$.
It suffices to add paths of length at most $\demand(u,v)=\dist{G}{\ell}(u,v) + \beta$ for $\{u,v\} \in K$.
Consequently, \AG achieves the aforementioned guarantees for unit-length additive $+\beta$-spanner, for all constant integer $\beta \geq 2$.
Interestingly, our proof still works analogous to~\cite{knudsen2014}. %

\begin{proof}[Proof of \Cref{cor:AdaptedGreedy+beta}]
For $v \in V$, $\myDeg_G(v)$ denotes the degree of $v$ in graph $G$.
For a spanner $H=(V,E')$ of $G=(V,E)$, define
\[ v(H)= \sum_{u,v \in V} \max \{0, d_G^\ell(u,v) - d_H^\ell(u,v) + \newStuff{(\beta +3)}\} \quad \text{ and } \quad c(H)= \sum_{v \in V} (\myDeg_H(v))^2.\]
Clearly, $0 \leq v(H) \leq \newStuff{(\beta +3)}n^2$ and by Cauchy Schwartz's inequality $(c(H) n)^{1/2} \geq 2 |E'|$.
The goal is to show that upon termination of \AG $c(H) \in \mathcal{O}(n^2)$, since this implies that the size $|E'|$, and consequently the weight ratio, is in $\mathcal{O}(n^{3/2})$.
We do this by proving that in each step of Phase 2 $c(H) - 12 v(H)$ does not increase.
Then, since $v(H) \in \mathcal{O}(n^2)$ this implies $c(H) \in \mathcal{O}(n^2)$, concluding the proof.

Consider a step in Phase 2 where $(u,v) \in K$ is processed. Let $H_0$ be the current spanner.
Assume $(u,v)$ violates the distance demand, i.e.,  $\distL{H_0}(u,v) > \demand(u,v)=\distL{G}(u,v) + \newStuff{\beta}$, as otherwise nothing happens.
Then, new edges are added to $H_0$ on a shortest $uv$-path $P$ in \newStuff{$\lowG$}. 
Let $H=H_0 \cup P$ the spanner after adding $P$. Let $P$ consist of nodes $u=w_0, w_1, \dots, w_{t-1}, w_t=v$ and $|P|=t$.
As the degree of any $w_i$ increases by at most $2$ in $H$, we get
\begin{gather}
    c(H)-c(H_0) \leq \sum_{i=0}^{t} \left((\myDeg_H(w_i))^2 - (\myDeg_H(w_i)-2)^2 \right)\leq 4 \sum_{i=0}^{t} \myDeg_H(w_i). \label{eq:cH-cH0}
\end{gather}
No node can be adjacent (in \newStuff{$\lowG$}) to more than $3$ nodes on $P$, since otherwise $P$ would not be a shortest $uv$-path in \newStuff{$\lowG$}.
Let $A$ be the set of nodes adjacent in $H$ to a node of $P$. Note that all nodes on $P$ are in $A$.
Since $E' \subseteq \newStuff{\lowE}$, we can bound $|A| \geq \frac{1}{3} \sum_{i=0}^{t} \myDeg_H(w_i)$.
Recall that $|P|=\distL{\newStuff{\lowG}}(u,v) \leq \demand(u,v) = \dist{G}{\ell}(u,v) +\newStuff{\beta}$ for all $\{u,v\} \in K$. 
Now, let $z \in A$.
Clearly, $\dist{H}{\ell}(u,z)+\dist{H}{\ell}(z,v) \leq |P|+2 = \dist{\newStuff{\lowG}}{\ell}(u,v) + 2 \leq \dist{G}{\ell}(u,v)+\newStuff{\beta+2}$.
Since the distance demand of $(u,v)$ is violated in $H_0$, we have $\distL{H_0}(u,z) + \distL{H_0}(z,v) > \demand(u,v) + 2 = \distL{G}(u,v) + \newStuff{\beta + 2}$.
Hence, 
\begin{gather}
    \distL{H}(u,z)+\distL{H}(z,v) < \distL{H_0}(u,z) + \distL{H_0}(z,v) \label{eq:H<H0}
\end{gather}
with either $\distL{H}(u,z) < \distL{H_0}(u,z)$ or $\distL{H}(z,v) < \distL{H_0}(z,v)$ or both.
Now, let $w_i$ be adjacent to~$z$ \newStuff{and recall that $\{u,z\}, \{z,v\} \in K = \binom{V}{2}$ with $\demand(u,z) = \distL{G}(u,z) + \beta$ and $\demand(z,v) = \distL{G}(z,v) + \beta$, respectively}. Then
\begin{align} 
    \distL{H}(u,z) &\leq \distL{H}(u,w_i) + \distL{H}(w_i, z) = \distL{\newStuff{\lowG}}(u,w_i) +1 \notag\\
                   &\leq \distL{\newStuff{\lowG}}(u,z) + \distL{\newStuff{\lowG}}(z,w_i) + 1 = \distL{\newStuff{\lowG}}(u,z) + 2 \leq \distL{G}(u,z) + \newStuff{\beta + 2}. \label{eq:H<G+4}
\end{align}
And analogously $\distL{H}(z,v) \leq \distL{G}(z,v) + \newStuff{\beta +2}$. 
The key observation is that the differences between $v(H) - v(H_0)$ %
can be based on nodes $z \in A$, as only they are affected by the inclusion of $P$ to $H_0$.
By \eqref{eq:H<H0}, we know that for every $z \in A$ the distance to either $u$ or $v$ must be strictly shorter in $H$ than it is in $H_0$.
Combined with \eqref{eq:H<G+4} yields
\begin{gather}
    \sum_{w \in V}  \max \{0, \distL{G}(z,w) -  \distL{H}(z,w) + \newStuff{\beta + 3} \}> \sum_{w \in V}  \max \{0,  \distL{G}(z,w) -  \distL{H_0}(z,w) + \newStuff{\beta + 3} \} \text{,} \label{eq:diff}
\end{gather}
where only summands $w \in \{u,v\} \subseteq V$ make a real difference.
As \eqref{eq:diff} holds for every $z \in A$, we get $v(H)-v(H_0) \geq \frac{1}{3} \sum_{i=0}^{t} \myDeg_H(w_i)$.
Combining this with \eqref{eq:cH-cH0} concludes the proof as
$(c(H) - 12v(H)) - (c(H_0) - 12v(H_0)) \leq 0$.
\end{proof}

\end{document}